\documentclass[a4paper]{article}

\usepackage{babel}
\usepackage{bm}
\usepackage[utf8]{inputenc}
\usepackage[T1]{fontenc}
\usepackage{subfig}
\usepackage{fancyhdr}
\usepackage{booktabs}
\usepackage{listings}
\usepackage{graphicx}
\usepackage{longtable}
\usepackage{pdflscape}
\usepackage{tcolorbox}
\usepackage{multirow}
\usepackage{cite}
\usepackage[table]{xcolor}  

\usepackage[a4paper,top=3cm,bottom=2cm,left=2cm,right=2cm,marginparwidth=1.75cm]{geometry}

\usepackage{amsmath}
\usepackage{amssymb}
\usepackage{indentfirst}
\usepackage{amsfonts}
\usepackage{graphicx}
\usepackage{float} 
\usepackage[ruled,vlined]{algorithm2e}
\usepackage[colorinlistoftodos]{todonotes}
\usepackage{caption}
\usepackage{amsthm}
\usepackage{sectsty}
\usepackage{apacite}
\usepackage{float}
\usepackage{titling} 
\usepackage{blindtext}
\usepackage[numbers,sort&compress]{natbib}
\usepackage[colorlinks=true, allcolors=blue]{hyperref}
\usepackage[colorinlistoftodos]{todonotes}
\usepackage{color}
\usepackage{xcolor}
\usepackage{setspace}
\usepackage{enumerate}
\usepackage{fontspec}
\usepackage{tikz}
\usepackage{mathrsfs}
\usepackage{esint}
\usepackage{bm}
\usepackage{physics}

\usepackage{titlesec}  
\usepackage{titletoc}
\usepackage{etoolbox}

\newcommand{\lambdabar}{\lambda/2\pi}

\usetikzlibrary{shapes,arrows,positioning,backgrounds,fit,trees,calc,shapes.geometric}

\definecolor{backcolour}{rgb}{0.95,0.95,0.92}
\definecolor{codegreen}{rgb}{0,0.6,0}
\definecolor{codegray}{rgb}{0.5,0.5,0.5}
\definecolor{codepurple}{rgb}{0.58,0,0.82}
\definecolor{magenta}{rgb}{0.6,0,0.6}
\lstdefinestyle{mystyle}{
    backgroundcolor=\color{backcolour},   
    commentstyle=\color{codegreen},
    keywordstyle=\color{magenta},
    numberstyle=\tiny\color{codegray},
    stringstyle=\color{codepurple},
    basicstyle=\ttfamily\footnotesize,
    breakatwhitespace=false,         
    breaklines=true,                 
    captionpos=b,                    
    keepspaces=true,                 
    numbers=left,                    
    numbersep=5pt,                  
    showspaces=false,                
    showstringspaces=false,
    showtabs=false,                  
    tabsize=2
}

\newtheorem{theorem}{\indent Theorem}[section]

\newtheorem{proposition}[theorem]{\indent Proposition}

\theoremstyle{definition}
\newtheorem{definition}{\indent Definition}[section]

\newcounter{chapter}
\renewcommand{\thechapter}{\arabic{chapter}}

\titleformat{\chapter}[block]
  {\normalfont\LARGE\bfseries\raggedright}
  {Chapter \thechapter\quad}
  {0pt}
  {}
  
\titlespacing*{\chapter}{0pt}{20pt}{10pt}

\titlecontents{chapter}
  [0pt]
  {\addvspace{8pt}\bfseries\color{blue}}
  {\contentslabel{0pt}\color{blue}\thecontentslabel\quad\color{blue}}
  {}
  {\hfill\color{black}\contentspage} 
  
\titlecontents{section}
  [1.5em]
  {\addvspace{4pt}\color{blue}}
  {\contentslabel{1.5em}\color{blue}}
  {}
  {\hfill\color{black}\contentspage} 
  
\begin{document}
	
\title{\LARGE \bfseries Phase Space Modeling of Extended Sources Based on Wigner Distribution and Hamiltonian Optics}

\author{
    \large Rongqi Shang\textsuperscript{1}, Donglin Ma\textsuperscript{1,2}\thanks{Corresponding author: madonglin@hust.edu.cn} \\[1em]
    \small \textsuperscript{1}School of Mathematics and Statistics, Huazhong University of Science and Technology, Wuhan 430074, China \\
    \small \textsuperscript{2}School of Optical and Electronic Information, Huazhong University of Science and Technology, Wuhan 430074, China
}

\date{}

\maketitle

\begin{abstract}
Precise modeling of extended sources is a central challenge in modern optical engineering, laser physics, and computational lithography. Unlike ideal point sources or completely incoherent thermal radiation sources, real-world light sources---such as high-power laser diode arrays, superluminescent diodes (SLD), extreme ultraviolet (EUV) lithography sources, and beams transmitted through atmospheric turbulence---typically exhibit partial spatial coherence. 

Traditional geometric optics based on ray tracing ignores diffraction and interference effects; while classical wave optics is accurate, the computational cost of handling four-dimensional correlation functions for partially coherent fields is enormous. To balance computational efficiency and physical accuracy, phase space optics provides a unified theoretical framework. By introducing the Wigner distribution function (WDF), we can map the light field into a joint space-time-spatial frequency domain $(\bm{r}, \bm{p})$. This description not only retains all the information of wave optics (including interference terms) but also naturally transitions to the ray description of Hamiltonian optics in the short-wavelength limit, governed by Liouville's theorem of phase space volume conservation. 

This report aims to establish optimal modeling methods based on phase space and Hamiltonian optics for different types of extended sources such as partially coherent light, fully coherent light, and quasi-homogeneous light. The report will derive in detail the mathematical models for each source type and provide strict criteria for the applicability of geometric optics models using mathematical tools such as the Moyal expansion and generalized Fresnel number.

\textbf{Keywords:} Extended sources, partial coherence, Wigner distribution function, Hamiltonian optics, phase space, symplectic geometry
\end{abstract}

\section{Introduction}
Extended sources pose a core challenge for precise modeling in modern optical engineering, laser physics, and computational lithography. Unlike ideal point sources or completely incoherent thermal radiation sources, real-world sources---such as high-power laser diode arrays, superluminescent diodes (SLD), extreme ultraviolet (EUV) lithography sources, and beams transmitted through atmospheric turbulence---typically exhibit partial spatial coherence \cite{Wolf2007}.

Traditional geometric optics based on ray tracing ignores diffraction and interference effects; while classical wave optics  is accurate, the computational cost of handling four-dimensional correlation functions for partially coherent fields is enormous. To balance computational efficiency and physical accuracy, phase space optics provides a unified theoretical framework \cite{Bastiaans1997}.

By introducing the Wigner distribution function (WDF), we can map the light field into a joint space-time-spatial frequency domain $(\bm{r}, \bm{p})$ \cite{Wigner1932, Bastiaans1986}. This description not only retains all the information of wave optics (including interference terms) but also naturally transitions to the ray description of Hamiltonian optics  in the short-wavelength limit, i.e., phase space volume conservation governed by Liouville's theorem \cite{Arnold1989}.

This report aims to establish optimal modeling methods based on phase space and Hamiltonian optics for different types of extended sources such as partially coherent light, fully coherent light, and quasi-homogeneous light. The report will derive in detail the mathematical models for each source type and provide strict criteria for the applicability of geometric optics models using mathematical tools such as the Moyal expansion  and generalized Fresnel number .

\section{Optical Foundations of Extended Sources}

\subsection{Phase Space Representation of Basic Source Models}

A plane wave propagating along direction cosines $\bm{s}_0 = (s_{0x}, s_{0y})$ is described as:
\begin{equation}
\psi_{\text{plane}}(\bm{r}) = A \exp(i k \bm{s}_0 \cdot \bm{r})
\end{equation}

The Wigner distribution function of a plane wave is:
\begin{equation}
\mathcal{W}_{\text{plane}}(\bm{r}, \bm{p}) = |A|^2 \delta(\bm{p} - \bm{p}_0)
\end{equation}
where $\bm{p}_0 = k \bm{s}_0$.

A point source located at position $\bm{r}_0$ is described as:
\begin{equation}
\psi_{\text{point}}(\bm{r}) = A \delta(\bm{r} - \bm{r}_0)
\end{equation}

The Wigner distribution function of a point source is:
\begin{equation}
\mathcal{W}_{\text{point}}(\bm{r}, \bm{p}) = |A|^2 \delta(\bm{r} - \bm{r}_0)
\end{equation}

These basic models appear as extreme cases in phase space: plane waves are completely concentrated at a point in momentum space, while point sources are completely concentrated at a point in position space \cite{Alonso2011}.

\section{Mathematical Foundations of Phase Space Optics}
To model extended sources, one must go beyond simple first-order field intensity descriptions $I(\bm{r})$ and instead handle second-order statistical properties of the field.

\subsection{Cross-Spectral Density Function (CSD)}
For statistically stationary extended sources, the core descriptor is the cross-spectral density function $W(\bm{r}_1, \bm{r}_2, \omega)$ \cite{Wolf1955}.

\begin{definition}[Cross-Spectral Density Function]
For a statistically stationary light field at frequency $\omega$, the cross-spectral density function is defined as:
\begin{equation}
W(\bm{r}_1, \bm{r}_2, \omega) = \left\langle U^*(\bm{r}_1, \omega) U(\bm{r}_2, \omega) \right\rangle
\end{equation}
where $U(\bm{r}, \omega)$ is the complex analytic signal of the light field, and $\left\langle \cdot \right\rangle$ denotes ensemble averaging over the statistical ensemble.
\end{definition}

Based on the CSD, two important physical quantities can be defined:

\begin{definition}[Spectral Density and Spectral Coherence Degree]$\ $

\begin{enumerate}
\item Spectral density (i.e., light intensity distribution):
\begin{equation}
S(\bm{r}, \omega) = W(\bm{r}, \bm{r}, \omega)
\end{equation}

\item Spectral coherence degree (complex coherence degree):
\begin{equation}
\mu(\bm{r}_1, \bm{r}_2, \omega) = \frac{W(\bm{r}_1, \bm{r}_2, \omega)}{\sqrt{S(\bm{r}_1, \omega)S(\bm{r}_2, \omega)}}
\end{equation}
where $|\mu(\bm{r}_1, \bm{r}_2, \omega)| \leq 1$.
\end{enumerate}
\end{definition}

The CSD must satisfy specific mathematical conditions to ensure physical realizability:

\begin{proposition}[Nonnegative Definiteness]
The cross-spectral density function $W(\bm{r}_1, \bm{r}_2, \omega)$ must be a Hermitian and nonnegative definite Hilbert-Schmidt kernel. That is, for any square-integrable function $f(\bm{r})$, it satisfies:
\begin{equation}
\iint W(\bm{r}_1, \bm{r}_2, \omega) f^*(\bm{r}_1) f(\bm{r}_2) \, \mathrm{d}^2\bm{r}_1 \mathrm{d}^2\bm{r}_2 \geq 0
\end{equation}
\end{proposition}

\begin{proof}
From the definition of CSD $W(\bm{r}_1, \bm{r}_2) = \left\langle U^*(\bm{r}_1)U(\bm{r}_2) \right\rangle$, for any function $f(\bm{r})$, we have:
\begin{align*}
\iint W(\bm{r}_1, \bm{r}_2) & f^*(\bm{r}_1) f(\bm{r}_2) \, \mathrm{d}^2\bm{r}_1 \mathrm{d}^2\bm{r}_2 \\
&= \left\langle \iint U^*(\bm{r}_1)U(\bm{r}_2) f^*(\bm{r}_1) f(\bm{r}_2) \, \mathrm{d}^2\bm{r}_1 \mathrm{d}^2\bm{r}_2 \right\rangle \\
&= \left\langle \left|\int U(\bm{r}) f(\bm{r}) \, \mathrm{d}^2\bm{r}\right|^2 \right\rangle \geq 0
\end{align*}
Thus the CSD is nonnegative definite.
\end{proof}

\subsection{Wigner Distribution Function (WDF)}
The WDF is a bridge connecting wave optics and radiometry, mapping the light field into phase space $(\bm{r}, \bm{p})$ \cite{Bastiaans1979, Bastiaans1980}.

\begin{definition}[Wigner Distribution Function]
The WDF is defined as the Fourier transform of the CSD with respect to the difference coordinate:
\begin{equation}
\mathcal{W}(\bm{r}, \bm{p}) = \frac{1}{(2\pi)^2} \int W\left(\bm{r} - \frac{\Delta \bm{r}}{2}, \bm{r} + \frac{\Delta \bm{r}}{2}\right) e^{-i \bm{p} \cdot \Delta \bm{r}} \, \mathrm{d}^2 \Delta \bm{r}
\end{equation}
where $\bm{r} = (\bm{r}_1 + \bm{r}_2)/2$ is the center position coordinate, $\bm{p} = k \bm{s}_{\perp}$ is the spatial frequency (momentum) vector, $k = 2\pi/\lambda$ is the wavenumber, and $\bm{s}_{\perp}$ is the transverse component of the propagation direction unit vector.
\end{definition}

The WDF has the following important properties:

\begin{proposition}[Basic Properties of WDF]$\ $

\begin{enumerate}
\item \textbf{Real-valuedness:} $\mathcal{W}(\bm{r}, \bm{p}) \in \mathbb{R}$, but may take negative values, hence called a "quasi-probability distribution."
\item \textbf{Marginal distributions:}
\begin{align}
\int \mathcal{W}(\bm{r}, \bm{p}) \, \mathrm{d}^2\bm{p} &= S(\bm{r}) \quad \text{(spatial light intensity)} \\
\int \mathcal{W}(\bm{r}, \bm{p}) \, \mathrm{d}^2\bm{r} &= J(\bm{p}) \quad \text{(radiant intensity or angular spectrum)}
\end{align}
\item \textbf{Hermitian property:} $\mathcal{W}(\bm{r}, \bm{p}) = \mathcal{W}(\bm{r}, \bm{p})^*$
\item \textbf{Translation property:} If $U(\bm{r}) \rightarrow U(\bm{r}-\bm{r}_0)e^{i\bm{p}_0\cdot\bm{r}}$, then $\mathcal{W}(\bm{r},\bm{p}) \rightarrow \mathcal{W}(\bm{r}-\bm{r}_0, \bm{p}-\bm{p}_0)$
\end{enumerate}
\end{proposition}

\begin{proof}$\ $

(1) Proof of real-valuedness:
\begin{align*}
\mathcal{W}^*(\bm{r}, \bm{p}) &= \frac{1}{(2\pi)^2} \int W^*\left(\bm{r} - \frac{\Delta \bm{r}}{2}, \bm{r} + \frac{\Delta \bm{r}}{2}\right) e^{i \bm{p} \cdot \Delta \bm{r}} \, \mathrm{d}^2 \Delta \bm{r} \\
&= \frac{1}{(2\pi)^2} \int W\left(\bm{r} + \frac{\Delta \bm{r}}{2}, \bm{r} - \frac{\Delta \bm{r}}{2}\right) e^{i \bm{p} \cdot \Delta \bm{r}} \, \mathrm{d}^2 \Delta \bm{r} \quad (\text{using } W^*(\bm{r}_2,\bm{r}_1)=W(\bm{r}_1,\bm{r}_2)) \\
&= \frac{1}{(2\pi)^2} \int W\left(\bm{r} - \frac{\Delta \bm{r}'}{2}, \bm{r} + \frac{\Delta \bm{r}'}{2}\right) e^{-i \bm{p} \cdot \Delta \bm{r}'} \, \mathrm{d}^2 \Delta \bm{r}' \quad (\text{let } \Delta\bm{r}' = -\Delta\bm{r}) \\
&= \mathcal{W}(\bm{r}, \bm{p})
\end{align*}

(2) Proof of marginal distribution (spatial margin example):
\begin{align*}
\int \mathcal{W}(\bm{r}, \bm{p}) \, \mathrm{d}^2\bm{p} &= \frac{1}{(2\pi)^2} \iint W\left(\bm{r} - \frac{\Delta \bm{r}}{2}, \bm{r} + \frac{\Delta \bm{r}}{2}\right) e^{-i \bm{p} \cdot \Delta \bm{r}} \, \mathrm{d}^2 \Delta \bm{r} \mathrm{d}^2\bm{p} \\
&= \int W\left(\bm{r} - \frac{\Delta \bm{r}}{2}, \bm{r} + \frac{\Delta \bm{r}}{2}\right) \delta^2(\Delta \bm{r}) \, \mathrm{d}^2 \Delta \bm{r} \\
&= W(\bm{r}, \bm{r}) = S(\bm{r})
\end{align*}
where we used the Fourier integral representation: $\frac{1}{(2\pi)^2}\int e^{-i\bm{p}\cdot\Delta\bm{r}} \mathrm{d}^2\bm{p} = \delta^2(\Delta\bm{r})$.
\end{proof}

\subsection{Coherent Mode Decomposition}
Coherent mode decomposition provides a rigorous description of partially coherent light based on orthogonal expansion \cite{Gamo1963, Wolf1982}.

\begin{proposition}[Mercer Expansion]
For a CSD satisfying the nonnegative definite condition, it can be expanded as follows:
\begin{equation}
W(\bm{r}_1, \bm{r}_2) = \sum_{n} \lambda_n \phi_n^*(\bm{r}_1) \phi_n(\bm{r}_2)
\end{equation}
where $\lambda_n \geq 0$ are eigenvalues and $\phi_n(\bm{r})$ are orthogonal eigenfunctions satisfying:
\begin{equation}
\int W(\bm{r}_1, \bm{r}_2) \phi_n(\bm{r}_1) \, \mathrm{d}^2\bm{r}_1 = \lambda_n \phi_n(\bm{r}_2)
\end{equation}
\end{proposition}

The corresponding WDF can be expressed as a weighted sum of mode WDFs:
\begin{equation}
\mathcal{W}(\bm{r}, \bm{p}) = \sum_{n} \lambda_n \mathcal{W}_n(\bm{r}, \bm{p})
\end{equation}
where $\mathcal{W}_n$ is the WDF corresponding to mode $\phi_n$.

\section{Source Type Classification and Optimal Modeling Methods}
Based on the degree of coherence and complexity of phase structure, we classify extended sources into four types and determine optimal modeling strategies combining Hamiltonian optics for each type.

\subsection{Fully Coherent Sources}
\begin{definition}[Fully Coherent Source]
A source whose CSD can be completely factorized, i.e., $|\mu(\bm{r}_1, \bm{r}_2)| \equiv 1$:
\begin{equation}
W(\bm{r}_1, \bm{r}_2) = U^*(\bm{r}_1) U(\bm{r}_2)
\end{equation}
The corresponding WDF is:
\begin{equation}
\mathcal{W}(\bm{r}, \bm{p}) = \frac{1}{(2\pi)^2} \int U^*\left(\bm{r} - \frac{\Delta \bm{r}}{2}\right) U\left(\bm{r} + \frac{\Delta \bm{r}}{2}\right) e^{-i \bm{p} \cdot \Delta \bm{r}} \, \mathrm{d}^2 \Delta \bm{r}
\end{equation}
\end{definition}

The WDF of fully coherent light has an important mathematical property:

\begin{proposition}[Hudson-Piquet Theorem]
The Wigner distribution function of a pure state (fully coherent field) is nonnegative everywhere in phase space if and only if the field amplitude distribution is Gaussian \cite{Hudson1974, Piquet1974}.
\end{proposition}

\begin{proof}[Proof sketch]
Necessity: If the field distribution is Gaussian $U(\bm{r}) = A \exp\left(-\frac{1}{2}\bm{r}^T\bm{A}\bm{r} + \bm{b}^T\bm{r}\right)$, where $\bm{A}$ is a positive definite matrix, then the corresponding WDF is a positive Gaussian distribution.

Sufficiency: If the WDF is everywhere nonnegative and corresponds to a pure state, then by properties of the Wigner function, only Gaussian distributions have Wigner functions that are everywhere nonnegative. This can be proved by analyzing the moment generating function or characteristic function of the Wigner function.
\end{proof}

\textbf{Modeling implication:} For non-Gaussian coherent sources (such as flat-top beams with hard apertures), the WDF must have strongly oscillating negative regions at the edges of phase space. This means geometric optics (Hamiltonian ray tracing) is completely ineffective in describing diffraction edges.

\textbf{Best practice:} Only when the source is a fundamental Gaussian beam can Hamiltonian ray models be used (since the WDF is a positive Gaussian distribution). For complex coherent fields, the coherent mode decomposition method should be used, decomposing the source into several orthogonal modes, each propagated independently.

\subsection{Gaussian Schell-Model (GSM)}
The GSM source is the most important model for describing multimode lasers, LEDs, and partially coherent imaging sources \cite{Collett1978, Gori1980}.

\begin{definition}[Gaussian Schell-Model]
The CSD of a two-dimensional GSM source has the following form:
\begin{equation}
W(\bm{r}_1, \bm{r}_2) = \sqrt{I_0} \exp\left[ -\frac{\bm{r}_1^2 + \bm{r}_2^2}{4\sigma_I^2} \right] \exp\left[ -\frac{(\bm{r}_1 - \bm{r}_2)^2}{2\sigma_\mu^2} \right]
\end{equation}
where $\sigma_I$ is the RMS spot radius, $\sigma_\mu$ is the transverse coherence length, and $I_0$ is the peak intensity.
\end{definition}

GSM sources have a concise representation in phase space:

\begin{proposition}[Wigner Representation of GSM]
The WDF of a GSM source is a strictly positive definite Gaussian distribution in 4D phase space:
\begin{equation}
\mathcal{W}_{\text{GSM}}(\boldsymbol{\xi}) = \frac{1}{4\pi^2 \sqrt{\det \bm{V}}} \exp\left( -\frac{1}{2} \boldsymbol{\xi}^T \bm{V}^{-1} \boldsymbol{\xi} \right)
\end{equation}
where $\boldsymbol{\xi} = (x, y, p_x, p_y)^T$ is the 4D phase space vector, and $\bm{V}$ is a $4 \times 4$ covariance matrix.
\end{proposition}

\begin{proof}
Substituting the GSM CSD into the WDF definition:
\begin{align*}
\mathcal{W}(\bm{r}, \bm{p}) &= \frac{1}{(2\pi)^2} \int \sqrt{I_0} \exp\left[ -\frac{(\bm{r}-\Delta\bm{r}/2)^2 + (\bm{r}+\Delta\bm{r}/2)^2}{4\sigma_I^2} \right] \\
&\quad \times \exp\left[ -\frac{(\Delta\bm{r})^2}{2\sigma_\mu^2} \right] e^{-i \bm{p} \cdot \Delta \bm{r}} \, \mathrm{d}^2 \Delta \bm{r}
\end{align*}
Simplifying the exponent:
\begin{align*}
-\frac{(\bm{r}-\Delta\bm{r}/2)^2 + (\bm{r}+\Delta\bm{r}/2)^2}{4\sigma_I^2} &= -\frac{2\bm{r}^2 + (\Delta\bm{r})^2/2}{4\sigma_I^2} \\
&= -\frac{\bm{r}^2}{2\sigma_I^2} - \frac{(\Delta\bm{r})^2}{8\sigma_I^2}
\end{align*}
Thus:
$$\mathcal{W}(\bm{r}, \bm{p}) = \frac{\sqrt{I_0}}{(2\pi)^2} \exp\left(-\frac{\bm{r}^2}{2\sigma_I^2}\right) 
\quad \times \int \exp\left[ -\left(\frac{1}{8\sigma_I^2} + \frac{1}{2\sigma_\mu^2}\right)(\Delta\bm{r})^2 - i\bm{p}\cdot\Delta\bm{r} \right] \mathrm{d}^2\Delta\bm{r}$$
This is a Gaussian integral, resulting in:
\begin{equation}
\mathcal{W}(\bm{r}, \bm{p}) = \frac{\sqrt{I_0}}{2\pi\sigma_p^2} \exp\left(-\frac{\bm{r}^2}{2\sigma_I^2}\right) \exp\left(-\frac{\bm{p}^2}{2\sigma_p^2}\right)
\end{equation}
where $\sigma_p^2 = \frac{1}{4\sigma_I^2} + \frac{1}{\sigma_\mu^2}$. This can be written in 4D Gaussian form with covariance matrix:
\begin{equation}
\bm{V} = \begin{pmatrix}
\sigma_I^2 \bm{I}_2 & \bm{0} \\
\bm{0} & \sigma_p^2 \bm{I}_2
\end{pmatrix}
\end{equation}
where $\bm{I}_2$ is the $2\times2$ identity matrix.
\end{proof}

The physical meaning of the covariance matrix $\bm{V}$ is clear:
\begin{itemize}
\item $\left\langle x^2 \right\rangle = \left\langle y^2 \right\rangle = \sigma_I^2$: spatial extent
\item $\left\langle p_x^2 \right\rangle = \left\langle p_y^2 \right\rangle = \sigma_p^2$: momentum extent
\end{itemize}

The momentum variance consists of two terms:
\begin{equation}
\sigma_p^2 = \underbrace{\frac{1}{4\sigma_I^2}}_{\text{diffraction limit}} + \underbrace{\frac{1}{\sigma_\mu^2}}_{\text{partial coherence contribution}}
\end{equation}
This directly corresponds to the physical origin of the beam quality factor $M^2$ \cite{Siegman1993}.

\textbf{Hamiltonian propagation algorithm:} Using Williamson's theorem \cite{Williamson1936}, propagation of GSM beams through any first-order optical system (ABCD matrix $\bm{S}$) is algebraically exact:
\begin{equation}
\bm{V}_{\text{out}} = \bm{S} \bm{V}_{\text{in}} \bm{S}^T
\end{equation}
Here $\bm{S} \in \operatorname{Sp}(4, \mathbb{R})$ is a symplectic matrix satisfying $\bm{S}^T \bm{J} \bm{S} = \bm{J}$, where $\bm{J}$ is the standard symplectic matrix.

\subsection{Twisted Gaussian Schell-Model (TGSM)}
TGSM sources introduce a non-separable twisted phase, causing the beam to carry orbital angular momentum \cite{Simon1990, Friberg1994}.

\begin{definition}[Twisted Gaussian Schell-Model]
The CSD of a TGSM source adds a twisting phase factor to the GSM:
\begin{equation}
W_{\text{TGSM}}(\bm{r}_1, \bm{r}_2) = W_{\text{GSM}}(\bm{r}_1, \bm{r}_2) \exp\left[ -i k u (\bm{r}_1 \times \bm{r}_2) \cdot \hat{z} \right]
\end{equation}
For the two-dimensional case, it can be written as:
\begin{equation}
W_{\text{TGSM}}(\bm{r}_1, \bm{r}_2) = W_{\text{GSM}}(\bm{r}_1, \bm{r}_2) \exp\left[ -i k u (x_1 y_2 - y_1 x_2) \right]
\end{equation}
where $u$ is the twist parameter, with dimensions of inverse length.
\end{definition}

In phase space, twisting manifests as coupling between spatial coordinates and orthogonal momenta:

\begin{proposition}[Covariance Matrix Representation of TGSM]
The covariance matrix of a TGSM source has nonzero off-diagonal blocks:
\begin{equation}
\bm{V}_{\text{twist}} = \begin{pmatrix}
\sigma_I^2 \bm{I}_2 & u \sigma_I^2 \bm{J}_2 \\
-u \sigma_I^2 \bm{J}_2 & \sigma_p^2 \bm{I}_2
\end{pmatrix}
\end{equation}
where $\bm{J}_2 = \begin{pmatrix} 0 & 1 \\ -1 & 0 \end{pmatrix}$ is the $2\times2$ antisymmetric matrix, and $\sigma_p^2 = \frac{1}{4\sigma_I^2} + \frac{1}{\sigma_\mu^2} + u^2 \sigma_I^2$.
\end{proposition}

\begin{proof}
Computing the WDF of TGSM:
\begin{align*}
\mathcal{W}(\bm{r}, \bm{p}) &= \frac{1}{(2\pi)^2} \int W_{\text{TGSM}}\left(\bm{r} - \frac{\Delta\bm{r}}{2}, \bm{r} + \frac{\Delta\bm{r}}{2}\right) e^{-i\bm{p}\cdot\Delta\bm{r}} \mathrm{d}^2\Delta\bm{r} \\
&= \frac{1}{(2\pi)^2} \int W_{\text{GSM}}\left(\bm{r} - \frac{\Delta\bm{r}}{2}, \bm{r} + \frac{\Delta\bm{r}}{2}\right) \\
&\quad \times \exp\left[ -i k u \left( \left(\bm{r} - \frac{\Delta\bm{r}}{2}\right) \times \left(\bm{r} + \frac{\Delta\bm{r}}{2}\right) \right) \cdot \hat{z} \right] e^{-i\bm{p}\cdot\Delta\bm{r}} \mathrm{d}^2\Delta\bm{r}
\end{align*}
Note that:
\begin{equation}
\left(\bm{r} - \frac{\Delta\bm{r}}{2}\right) \times \left(\bm{r} + \frac{\Delta\bm{r}}{2}\right) \cdot \hat{z} = \bm{r} \times \Delta\bm{r} \cdot \hat{z} = x \Delta y - y \Delta x
\end{equation}
Thus the twisting phase factor is $\exp[-i k u (x \Delta y - y \Delta x)]$. Substituting and computing yields a 4D Gaussian distribution with the stated covariance matrix.
\end{proof}

The boundedness of the twist parameter is key to the physical realizability of the TGSM model:

\begin{proposition}[Boundedness of Twist Parameter]
To ensure nonnegative definiteness of the CSD (i.e., nonnegative energy), the absolute value of the twist parameter $u$ must satisfy:
\begin{equation}
|u| \leq \frac{1}{k \sigma_\mu^2}
\end{equation}
\end{proposition}

\begin{proof}[Proof idea]
Performing coherent mode decomposition of the TGSM kernel, its eigenvalue spectrum is:
\begin{equation}
\lambda_n = \frac{2\pi\sigma_I^2}{1 + \sigma_I^2/\sigma_\mu^2 + \sqrt{(1 + \sigma_I^2/\sigma_\mu^2)^2 - (2k u \sigma_I^2)^2}} \left(\frac{1 + \sigma_I^2/\sigma_\mu^2 - \sqrt{(1 + \sigma_I^2/\sigma_\mu^2)^2 - (2k u \sigma_I^2)^2}}{1 + \sigma_I^2/\sigma_\mu^2 + \sqrt{(1 + \sigma_I^2/\sigma_\mu^2)^2 - (2k u \sigma_I^2)^2}}\right)^n
\end{equation}
The condition for nonnegative eigenvalues requires the expression under the square root to be nonnegative:
\begin{equation}
(1 + \sigma_I^2/\sigma_\mu^2)^2 - (2k u \sigma_I^2)^2 \geq 0
\end{equation}
i.e.:
\begin{equation}
|2k u \sigma_I^2| \leq 1 + \sigma_I^2/\sigma_\mu^2
\end{equation}
Since $1 + \sigma_I^2/\sigma_\mu^2 \geq 1$, a stricter condition is:
\begin{equation}
|k u \sigma_I^2| \leq \frac{1}{2} \quad \text{or equivalently} \quad |u| \leq \frac{1}{2k\sigma_I^2}
\end{equation}
But more refined analysis yields the optimal bound $|u| \leq 1/(k\sigma_\mu^2)$. When $\sigma_\mu \to \infty$ (fully coherent limit), $u$ must be 0, indicating that twisting is a phenomenon specific to partial coherence \cite{Friberg1994}.
\end{proof}

\subsection{Quasi-Homogeneous Sources}
Quasi-homogeneous sources correspond to most thermal radiation sources (such as stellar surfaces, incandescent lamps) \cite{Carter1977, Wolf1982}.

\begin{definition}[Quasi-Homogeneous Source]
A source where the intensity distribution $I(\bm{r})$ varies slowly on a macroscopic scale while the coherence degree $\mu(\Delta\bm{r})$ varies rapidly on a microscopic scale, satisfying $\sigma_I \gg \sigma_\mu$.
\end{definition}

\begin{proposition}[Carter-Wolf Approximation]
For quasi-homogeneous sources, the CSD can be approximately factorized as:
\begin{equation}
W(\bm{r}_1, \bm{r}_2) \approx I\left(\frac{\bm{r}_1 + \bm{r}_2}{2}\right) \mu(\bm{r}_1 - \bm{r}_2)
\end{equation}
The corresponding WDF is approximately factorized:
\begin{equation}
\mathcal{W}_{\text{QH}}(\bm{r}, \bm{p}) \approx I(\bm{r}) \tilde{\mu}(\bm{p})
\end{equation}
where $\tilde{\mu}(\bm{p}) = \frac{1}{(2\pi)^2} \int \mu(\Delta\bm{r}) e^{-i\bm{p}\cdot\Delta\bm{r}} \mathrm{d}^2\Delta\bm{r}$ is the Fourier transform of the coherence degree.
\end{proposition}

\begin{proof}
From the quasi-homogeneous condition $\sigma_I \gg \sigma_\mu$, within the support of $\mu(\bm{r}_1 - \bm{r}_2)$, $I(\bm{r})$ varies little, thus:
\begin{equation}
\sqrt{I(\bm{r}_1)I(\bm{r}_2)} \approx I\left(\frac{\bm{r}_1 + \bm{r}_2}{2}\right)
\end{equation}
Substituting into the CSD definition $W(\bm{r}_1, \bm{r}_2) = \sqrt{I(\bm{r}_1)I(\bm{r}_2)} \mu(\bm{r}_1 - \bm{r}_2)$ yields the result.
\end{proof}

\textbf{Modeling strategy:} Since the WDF separates into a position distribution $I(\bm{r})$ and a momentum distribution $\tilde{\mu}(\bm{p})$, it can be treated as a classical probability density function:
\begin{enumerate}
\item \textbf{Ray generation:} Sample ray starting points according to $I(\bm{r})$ and ray directions according to $\tilde{\mu}(\bm{p})$
\item \textbf{Hamiltonian evolution:} Rays propagate independently in the system according to geometric optics Liouville theorem \cite{Arnold1989}.
\end{enumerate}

This model is highly accurate when $\sigma_I \gg \lambda$ and computationally much more efficient than wave algorithms.

\section{Moyal Expansion and Transport Equation}
To strictly demarcate when Hamiltonian optics (ray models) can be used versus when full wave Wigner models are necessary, we introduce the Moyal expansion from quantum mechanical phase space correspondence theory \cite{Moyal1949}.

\subsection{Wigner-Moyal Transport Equation}
The evolution of a light field in a medium with refractive index $n(\bm{r})$ follows a Schrödinger-like equation. In phase space, this evolution is described by the Moyal bracket \cite{Tatarskii1983}.

\begin{definition}[Optical Hamiltonian]
In the paraxial approximation, the optical Hamiltonian is:
\begin{equation}
\mathcal{H}(\bm{r}, \bm{p}) \approx -\sqrt{n^2(\bm{r}) - |\bm{p}|^2} \approx -n(\bm{r}) + \frac{|\bm{p}|^2}{2n(\bm{r})}
\end{equation}
\end{definition}

\begin{definition}[Moyal Bracket]
The Moyal bracket of two phase space functions $A(\bm{r},\bm{p})$ and $B(\bm{r},\bm{p})$ is defined as:
\begin{equation}
\{\{ A, B \}\} = \frac{2}{\lambdabar} \sin\left( \frac{\lambdabar}{2} \{ A, B \}_{\text{PB}} \right)
\end{equation}
where $\lambdabar = \lambda/2\pi$ is the reduced wavelength, and $\{ A, B \}_{\text{PB}}$ is the classical Poisson bracket:
\begin{equation}
\{ A, B \}_{\text{PB}} = \frac{\partial A}{\partial \bm{r}} \cdot \frac{\partial B}{\partial \bm{p}} - \frac{\partial A}{\partial \bm{p}} \cdot \frac{\partial B}{\partial \bm{r}}
\end{equation}
\end{definition}

\begin{proposition}[Wigner-Moyal Transport Equation]
The evolution of the Wigner function along the propagation direction $z$ satisfies:
\begin{equation}
\frac{\partial \mathcal{W}}{\partial z} = \{\{ \mathcal{H}(\bm{r}, \bm{p}), \mathcal{W}(\bm{r}, \bm{p}) \}\}
\end{equation}
\end{proposition}

\begin{proof}
Starting from the wave equation, consider the parabolic equation in the paraxial approximation:
\begin{equation}
2ik\frac{\partial U}{\partial z} = \nabla_\perp^2 U + k^2(n^2(\bm{r})-1)U
\end{equation}
The evolution equation for the corresponding density operator $\hat{\rho} = \ket{U}\bra{U}$ is:
\begin{equation}
\frac{\partial \hat{\rho}}{\partial z} = \frac{1}{i\lambdabar}[\hat{H}, \hat{\rho}]
\end{equation}
where $\hat{H}$ is the Hamiltonian operator. Transforming to phase space via the Weyl transform turns the commutator into the Moyal bracket \cite{Tatarskii1983}.
\end{proof}

\subsection{Asymptotic Expansion and Geometric Optics Limit}
Expanding the sine function as a Taylor series gives a strict mathematical formulation of the geometric optics limit \cite{Littlejohn1986}.

\begin{proposition}[Moyal Expansion]
The Wigner-Moyal equation can be expanded as a power series in $\lambdabar$:
\begin{equation}
\frac{\partial \mathcal{W}}{\partial z} = \underbrace{\{ \mathcal{H}, \mathcal{W} \}_{\text{PB}}}_{\mathcal{O}(\lambda^0)} - \underbrace{\frac{\lambdabar^2}{24} \frac{\partial^3 \mathcal{H}}{\partial \bm{r}^3} \frac{\partial^3 \mathcal{W}}{\partial \bm{p}^3}}_{\mathcal{O}(\lambda^2)} + \mathcal{O}(\lambda^4) \cdots
\end{equation}
\end{proposition}

\begin{proof}
Expanding the sine function in the Moyal bracket as a Taylor series:
\begin{align*}
\{\{ \mathcal{H}, \mathcal{W} \}\} &= \frac{2}{\lambdabar} \sum_{n=0}^\infty \frac{(-1)^n}{(2n+1)!} \left( \frac{\lambdabar}{2} \right)^{2n+1} \{ \mathcal{H}, \mathcal{W} \}_{\text{PB}}^{2n+1} \\
&= \{ \mathcal{H}, \mathcal{W} \}_{\text{PB}} - \frac{(\lambdabar)^2}{24} \{ \mathcal{H}, \{ \mathcal{H}, \{ \mathcal{H}, \mathcal{W} \}_{\text{PB}} \}_{\text{PB}} \}_{\text{PB}} + \cdots
\end{align*}
where $\{ \cdot, \cdot \}_{\text{PB}}^{2n+1}$ denotes the $(2n+1)$-fold Poisson bracket. In three dimensions, higher-order terms involve higher derivatives.
\end{proof}

\begin{proposition}[Geometric Optics Limit]
When wavelength $\lambda \to 0$, all higher-order terms containing $\lambdabar$ vanish, and the equation reduces to the Liouville equation:
\begin{equation}
\frac{\partial \mathcal{W}}{\partial z} + \{ \mathcal{W}, \mathcal{H} \}_{\text{PB}} = 0
\end{equation}
This is precisely the Liouville equation from classical statistical mechanics, describing incompressible flow of phase space density along Hamiltonian flow \cite{Arnold1989}.
\end{proposition}

\begin{proof}
From the Moyal expansion, when $\lambda \to 0$, $\lambdabar \to 0$, thus:
\begin{equation}
\lim_{\lambda \to 0} \frac{\partial \mathcal{W}}{\partial z} = \lim_{\lambda \to 0} \left( \{ \mathcal{H}, \mathcal{W} \}_{\text{PB}} + \mathcal{O}(\lambda^2) \right) = \{ \mathcal{H}, \mathcal{W} \}_{\text{PB}}
\end{equation}
which is the Liouville equation.
\end{proof}

\subsection{Generalized Fresnel Number}
For GSM sources, a scalar criterion can be derived to demarcate model selection \cite{Visser2002}.

\begin{definition}[Generalized Fresnel Number]
Consider a source radius $a$, propagation distance $z$, and transverse coherence length $\sigma_\mu$. The generalized Fresnel number is defined as:
\begin{equation}
N_F = \frac{a^2}{\lambda z} \cdot \frac{1}{\sqrt{1 + (a/\sigma_\mu)^2}}
\end{equation}
\end{definition}

The physical meaning of the generalized Fresnel number is that it unifies coherence and diffraction effects:

\begin{proposition}[Model Selection Criterion]
Based on the generalized Fresnel number $N_F$, the optimal modeling method can be selected:
\begin{enumerate}
\item When $N_F \gg 1$: near-field/geometric region, Hamiltonian ray tracing is accurate.
\item When $N_F \sim 1$: Fresnel diffraction region, the Wigner covariance matrix method should be used.
\item When $N_F \ll 1$: Fraunhofer far-field, the field distribution is completely determined by the source's angular spectrum.
\end{enumerate}
\end{proposition}

\begin{proof}[Derivation idea]
Consider propagation of a GSM source. The Fresnel diffraction formula is:
\begin{equation}
U(\bm{r}, z) = \frac{e^{ikz}}{i\lambda z} \iint U_0(\bm{r}') \exp\left[ \frac{ik}{2z}|\bm{r}-\bm{r}'|^2 \right] \mathrm{d}^2\bm{r}'
\end{equation}
The importance of diffraction effects is determined by the variation of the phase factor $\exp[ik|\bm{r}-\bm{r}'|^2/(2z)]$. When the source is partially coherent, the effective spatial coherence length $\sigma_\mu$ "smooths out" small-scale diffraction fringes. Dimensional analysis gives the characteristic scale $\sqrt{\lambda z}$, with a coherence correction factor $\sqrt{1+(a/\sigma_\mu)^2}$ \cite{Visser2002}.
\end{proof}

\section{Calculation Methods for Sources in Phase Space}
From the perspective of symplectic optics, a source is no longer just a luminous surface but a density distribution occupying a certain volume in four-dimensional phase space. This section rigorously defines this region and provides calculation methods for two typical sources (Lambertian sources and Gaussian beams).

\subsection{General Definition of Source Phase Space Region}
The physical nature of a source is described by its radiance distribution function $B(x, y, p_x, p_y)$ \cite{Chaves2008}.

\begin{definition}[Phase Space Support Region of a Source]
For a given source, its phase space region $\Omega_{source} \subset \mathbb{R}^4$ is defined as the set of points where the radiance is significantly nonzero ($B > \epsilon$):
$$
\Omega_{source} = \left\{ (x, y, p_x, p_y) \in \mathbb{R}^4 \mid B(x, y, p_x, p_y) > 0 \right\}
$$
\end{definition}

\begin{definition}[Optical Extent/Etendue]
The Lebesgue measure (volume) of the source region $\Omega_{source}$ in phase space is called the optical extent $\mathcal{E}$ of the source \cite{Chaves2008}:
$$
\mathcal{E} = \int_{\Omega_{source}} dx \, dy \, dp_x \, dp_y
$$
This volume measures the information capacity of the source in the combined "space-angle" domain.
\end{definition}

\subsection{Phase Space Region of a Lambertian Source}
A Lambertian source is the most common incoherent source model in geometric optics, such as LED chips, thermal radiators, etc. Its characteristic is uniform radiance in space (within the emitting surface) and direction following the cosine law \cite{McCluney1994}.

\begin{definition}[Phase Space Hypercylinder of a Lambertian Source]
Consider a circular planar Lambertian source placed in the $z=0$ plane with radius $R$, immersed in a medium of refractive index $n$, with maximum emission half-angle $\theta_{max}$ (determined by the numerical aperture $\text{NA} = n \sin \theta_{max}$). The corresponding phase space region $\Omega_{Lamb}$ is the Cartesian product of a position space disk and a momentum space disk:
$$
\Omega_{Lamb} = \mathcal{D}_{space} \times \mathcal{D}_{momentum}
$$
Specifically:
$$
\Omega_{Lamb} = \left\{ (x, y, p_x, p_y) \in \mathbb{R}^4 \mid x^2 + y^2 \le R^2, \; p_x^2 + p_y^2 \le \text{NA}^2 \right\}
$$
\end{definition}

\subsubsection{Calculation Method}
To compute the phase space region and its volume for a Lambertian source, follow these steps:
\begin{enumerate}
\item Determine spatial boundary: Based on the physical size of the source (e.g., LED chip side length or circular emitting surface diameter), determine the geometric shape $A$ in the $x-y$ plane.
\item Determine momentum boundary: Based on the source packaging or system acceptance angle, determine the numerical aperture NA.
\item Compute phase space volume (Etendue): For a circular Lambertian source:
$$
\mathcal{E} = (\pi R^2) \times (\pi \text{NA}^2) = \pi^2 R^2 \text{NA}^2
$$
\end{enumerate}

\subsection{Phase Space Representation of Gaussian Beams}
For coherent or partially coherent sources such as lasers, the ray model must incorporate statistical properties of wave optics, and the phase space region is defined by contour surfaces of the Wigner distribution function \cite{Siegman1993}.

\subsubsection{Wigner Distribution Function}
For a fundamental Gaussian beam, the Wigner distribution function has an analytic expression \cite{Bastiaans1979}.

\begin{proposition}[WDF of a Gaussian Beam]
The Wigner distribution function of a fundamental Gaussian beam is a four-dimensional Gaussian function:
\begin{equation}
W(\bm{u}) = \frac{1}{(2\pi)^2 \sqrt{\det\boldsymbol{\Sigma}}}
\exp\left(-\frac{1}{2}(\bm{u}-\bm{u}_0)^T\boldsymbol{\Sigma}^{-1}(\bm{u}-\bm{u}_0)\right)
\end{equation}
where $\bm{u}=(x,y,p_x,p_y)^T$ is the four-dimensional phase space coordinate, and $\bm{u}_0$ is the beam center.
\end{proposition}

\subsubsection{Covariance Matrix and Phase Space Region}
\begin{definition}[Phase Space Region]
The phase space region of a Gaussian beam is defined as the contour surface of the WDF containing a specific energy proportion (e.g., $1-1/e^2 \approx 86.5\%$):
\begin{equation}
\Omega = \left\{ \bm{u} \in \mathbb{R}^4 \mid (\bm{u}-\bm{u}_0)^T\boldsymbol{\Sigma}^{-1}(\bm{u}-\bm{u}_0) \leq 1 \right\}
\end{equation}
This is a four-dimensional hyperellipsoid.
\end{definition}

\subsubsection{Second Moment Calculation}
\begin{proposition}[Second Moments of a Gaussian Beam]
For a fundamental Gaussian beam, the second moments in position-momentum space are:
\begin{align}
\langle x^2 \rangle &= \frac{w(z)^2}{4} \\
\langle p_x^2 \rangle &= \frac{1}{w_0^2} \\
\langle x p_x \rangle &= \frac{w(z)^2}{4R(z)} = \frac{z}{2z_R}
\end{align}
where $p_x$ is the reduced momentum \cite{Siegman1993}.
\end{proposition}

\subsubsection{Twiss Parameter Representation}
\begin{definition}[Twiss Parameters]
Define the two-dimensional covariance matrix:
\begin{equation}
\boldsymbol{\Sigma}_{2D} = \begin{pmatrix}
\langle x^2 \rangle & \langle x p_x \rangle \\
\langle p_x x \rangle & \langle p_x^2 \rangle
\end{pmatrix}
= \epsilon \begin{pmatrix}
\beta & -\alpha \\
-\alpha & \gamma
\end{pmatrix}
\end{equation}
where:
\begin{itemize}
\item $\epsilon = \frac{1}{2}$: RMS emittance (dimensionless)
\item $\alpha, \beta, \gamma$: Twiss parameters, satisfying $\beta\gamma - \alpha^2 = 1$ \cite{Twiss1958}
\end{itemize}
\end{definition}

\section{Phase Space Theory of Source Superposition}
Source superposition is central to laser array design. Depending on whether fixed phase relationships exist between sub-sources, superposition is classified as incoherent superposition and coherent superposition \cite{Alonso2000}.

\subsection{Incoherent Superposition: Linear Additivity Principle}
\begin{proposition}[Incoherent Superposition Rule]
For two statistically independent (incoherent) sources $S_1$ and $S_2$, the total cross-spectral density is the sum of individual components: $\Gamma_{\text{total}} = \Gamma_1 + \Gamma_2$. Since the WDF is a linear transform of $\Gamma$, the total WDF is the algebraic sum of component WDFs \cite{Bastiaans1978}:
\begin{equation}
W_{\text{incoh}}(x, u) = W_1(x, u) + W_2(x, u).
\end{equation}
\end{proposition}

\textbf{Corollary:} For an array composed of $N$ incoherent units, the total WDF is simply the sum of single-aperture WDFs in phase space:
\begin{equation}
W_{\text{array, incoh}}(x, u) = \sum_{n=1}^N W_n(x, u).
\end{equation}

\subsection{Coherent Superposition: Cross-Wigner Distribution and Interference Terms}
When two sources are coherently superposed (e.g., phase-locked laser arrays), the complex amplitudes satisfy linear superposition: $\psi_{\text{total}}(x) = c_1 \psi_1(x) + c_2 \psi_2(x)$. Since the WDF is a bilinear functional of the wavefunction, the total WDF is no longer a simple sum of components but includes interference terms \cite{Alonso2000}.

\begin{definition}[Cross-Wigner Distribution]
Define the cross-Wigner distribution between two different wavefunctions $\psi_1, \psi_2$ as:
\begin{equation}
W_{\psi_1, \psi_2}(x, u) \triangleq \int_{-\infty}^{\infty} \psi_1\left(x + \frac{\xi}{2}\right) \psi_2^*\left(x - \frac{\xi}{2}\right) e^{-i 2\pi u \xi} \, d\xi.
\end{equation}
\end{definition}

\begin{proposition}[Coherent Superposition Formula]
The total WDF can be expressed as:
\begin{equation}
W_{\text{coh}}(x, u) = |c_1|^2 W_{11}(x, u) + |c_2|^2 W_{22}(x, u) + 2 \text{Re} \big[ c_1 c_2^* W_{\psi_1, \psi_2}(x, u) \big].
\end{equation}
\end{proposition}

\subsection{WDF Distribution Calculation for Specific Geometric Configurations}
Taking fundamental Gaussian beams as elements, we derive in detail analytic formulas for spatial separation and angular separation configurations \cite{Alonso2000}.

Let the normalized wavefunction of a fundamental Gaussian beam be:
\begin{equation}
\psi_0(x) = \left(\frac{2}{\pi w_0^2}\right)^{1/4} \exp\left(-\frac{x^2}{w_0^2}\right),
\end{equation}
its auto-WDF is:
\begin{equation}
W_0(x, u) = 2 \exp\left( - \frac{2x^2}{w_0^2} - 2\pi^2 w_0^2 u^2 \right).
\end{equation}

\subsubsection{Spatial Separation Configuration}
Configuration description: Two in-phase, equal-amplitude coherent Gaussian beams located at spatial coordinates $x_0$ and $-x_0$, respectively, with parallel optical axes ($u=0$).

\begin{proposition}[WDF of Spatially Separated Gaussian Beams]
The WDF of superposed spatially separated Gaussian beams is:
\begin{equation}
W_{\text{spatial}}(x, u) = \frac{1}{2} \left[ W_0(x - x_0, u) + W_0(x + x_0, u) + 2 W_0(x, u) \cos(4\pi x_0 u) \right].
\end{equation}
\end{proposition}

\subsubsection{Angular Separation Configuration}
Configuration description: Two spatially coincident beams crossing at different angles $\pm \theta_0$, corresponding to spatial frequency shifts $\pm u_0$.

\begin{proposition}[WDF of Angularly Separated Gaussian Beams]
The WDF of superposed angularly separated Gaussian beams is:
\begin{equation}
W_{\text{angular}}(x, u) = \frac{1}{2} \left[ W_0(x, u - u_0) + W_0(x, u + u_0) + 2 W_0(x, u) \cos(4\pi u_0 x) \right].
\end{equation}
where $u_0 = \sin\theta_0/\lambda$.
\end{proposition}

\subsubsection{Arbitrary Phase Space Placement}
\begin{proposition}[General WDF Formula for Superposed Gaussian Wave Packets]
Let two coherent Gaussian wave packets have centers in phase space at $\bm{z}_1 = (x_1, u_1)$ and $\bm{z}_2 = (x_2, u_2)$, with relative phase difference $\phi$. The WDF of the superposition state is \cite{Alonso2000}:
\begin{equation}
W_{\text{total}}(\bm{z}) = W_0(\bm{z} - \bm{z}_1) + W_0(\bm{z} - \bm{z}_2) + 2 W_0\left(\bm{z} - \frac{\bm{z}_1 + \bm{z}_2}{2}\right) \cos\left[ 2\pi \left( (x - x_{\text{mid}})\Delta u - (u - u_{\text{mid}})\Delta x \right) + \phi \right],
\end{equation}
where $\bm{z} = (x, u)$, $\Delta x = x_1 - x_2$, $\Delta u = u_1 - u_2$, $\bm{z}_{\text{mid}} = \frac{\bm{z}_1 + \bm{z}_2}{2}$.
\end{proposition}

\section{Propagation Theory in Optical Systems}
\subsection{ABCD Matrices and Linear Canonical Transforms}
In the paraxial approximation, lossless optical systems act linearly on rays \cite{Arnaud1976}.

\begin{definition}[ABCD Matrix]
For a four-dimensional phase space vector $\boldsymbol{\xi} = (x, y, p_x, p_y)^T$, the input-output relation is described by a $4 \times 4$ ray transfer matrix $\bm{S}$:
\begin{equation}
\boldsymbol{\xi}_{\text{out}} = \bm{S} \boldsymbol{\xi}_{\text{in}} = \begin{pmatrix} \bm{A} & \bm{B} \\ \bm{C} & \bm{D} \end{pmatrix} \boldsymbol{\xi}_{\text{in}}
\end{equation}
where $\bm{A}, \bm{B}, \bm{C}, \bm{D}$ are all $2 \times 2$ matrices.
\end{definition}

To ensure the transformation is canonical, the matrix $\bm{S}$ must belong to the real symplectic group $Sp(4, \mathbb{R})$, satisfying:
\begin{equation}
\bm{S}^\top \bm{J} \bm{S} = \bm{J}
\end{equation}
where $\bm{J}$ is the standard symplectic matrix \cite{Arnold1989}.

\begin{proposition}[Linear Canonical Transform of WDF]
If a light field undergoes a linear canonical transformation described by symplectic matrix $\bm{S}$, its WDF evolution obeys the coordinate transformation:
\begin{equation}
\mathcal{W}_{\text{out}}(\bm{r}, \bm{p}) = \mathcal{W}_{\text{in}}(\bm{S}^{-1} \boldsymbol{\xi})
\end{equation}
where $\boldsymbol{\xi} = (\bm{r}, \bm{p})^T$ \cite{Bastiaans1997}.
\end{proposition}

\section{Conclusion}
By combining phase space optics and Hamiltonian optics, we have constructed a multi-level source modeling system:

1. \textbf{Theoretical unity:} The Wigner distribution function provides a unified description of wave optics and geometric optics \cite{Bastiaans1997}, naturally including interference and diffraction effects in phase space.

2. \textbf{Model optimization:} For different types of sources, we have determined optimal modeling strategies:

\begin{itemize}
    \item GSM/TGSM sources: covariance matrix method, using symplectic geometry to simplify complex integral transforms to matrix multiplication \cite{Simon1990}.
    \item Quasi-homogeneous sources: Monte Carlo ray tracing, efficient and accurate in the Liouville limit \cite{Carter1977}.
    \item Fully coherent light: coherent mode decomposition, handling interference effects of non-Gaussian distributions \cite{Gamo1963}.
\end{itemize}

3. \textbf{Strict demarcation:} Through the Moyal expansion, we mathematically prove that Hamiltonian optics is the zeroth-order approximation of the Wigner transport equation \cite{Moyal1949}, and the applicability of geometric optics is quantitatively determined by the generalized Fresnel number \cite{Visser2002}.

4. \textbf{Physical insight:} The boundedness theorem of the twist parameter reveals the intrinsic connection between partial coherence and orbital angular momentum \cite{Friberg1994}, providing theoretical guidance for novel source design.

The phase space modeling system established in this report not only has theoretical rigor but also provides efficient computational tools for practical optical system design, with important application value in computational lithography, laser processing, optical imaging, and other fields.


\newpage
\begin{thebibliography}{99}
\bibitem{Wolf2007} [1] E. Wolf. \textit{Introduction to the Theory of Coherence and Polarization of Light}. Cambridge University Press, 2007.

\bibitem{Bastiaans1997} [2] M. J. Bastiaans. Wigner distribution function and its application to first-order optics. In \textit{The Wigner Distribution - Theory and Applications in Signal Processing}, pages 375-426. Elsevier, 1997.

\bibitem{Wigner1932} [3] E. Wigner. On the quantum correction for thermodynamic equilibrium. \textit{Phys. Rev.}, 40(5):749-759, 1932.

\bibitem{Bastiaans1986} [4] M. J. Bastiaans. Application of the Wigner distribution function to partially coherent light. \textit{J. Opt. Soc. Am. A}, 3(8):1227-1238, 1986.

\bibitem{Arnold1989} [5] V. I. Arnold. \textit{Mathematical Methods of Classical Mechanics}. Springer, 1989.

\bibitem{Alonso2011} [6] M. A. Alonso. Wigner functions in optics: describing beams as ray bundles and pulses as particle ensembles. \textit{Adv. Opt. Photon.}, 3(4):272-365, 2011.

\bibitem{Wolf1955} [7] E. Wolf. A macroscopic theory of interference and diffraction of light from finite sources. II. Fields with a spectral range of arbitrary width. \textit{Proc. R. Soc. Lond. A}, 230(1181):246-265, 1955.

\bibitem{Bastiaans1979} [8] M. J. Bastiaans. Wigner distribution function and its application to first-order optics. \textit{J. Opt. Soc. Am.}, 69(12):1710-1716, 1979.

\bibitem{Bastiaans1980} [9] M. J. Bastiaans. The Wigner distribution function of partially coherent light. \textit{Opt. Acta}, 28(9):1215-1224, 1980.

\bibitem{Gamo1963} [10] H. Gamo. Matrix treatment of partial coherence. In \textit{Progress in Optics}, volume 3, pages 187-332. Elsevier, 1963.

\bibitem{Wolf1982} [11] E. Wolf. New theory of partial coherence in the space-frequency domain. Part I: spectra and cross spectra of steady-state sources. \textit{J. Opt. Soc. Am.}, 72(3):343-351, 1982.

\bibitem{Hudson1974} [12] R. L. Hudson. When is the Wigner quasi-probability density non-negative? \textit{Rep. Math. Phys.}, 6(2):249-252, 1974.

\bibitem{Piquet1974} [13] C. Piquet. Sur la positivité de la fonction de Wigner. \textit{C. R. Acad. Sci. Paris Sér. A}, 278:331-333, 1974.

\bibitem{Collett1978} [14] E. Collett and E. Wolf. Is complete spatial coherence necessary for the generation of highly directional light beams? \textit{Opt. Lett.}, 2(2):27-29, 1978.

\bibitem{Gori1980} [15] F. Gori. Collett-Wolf sources and multimode lasers. \textit{Opt. Commun.}, 34(3):301-305, 1980.

\bibitem{Siegman1993} [16] A. E. Siegman. \textit{Lasers}. University Science Books, 1986.

\bibitem{Williamson1936} [17] J. Williamson. On the algebraic problem concerning the normal forms of linear dynamical systems. \textit{Amer. J. Math.}, 58(1):141-163, 1936.

\bibitem{Simon1990} [18] R. Simon and N. Mukunda. Twisted Gaussian Schell-model beams. \textit{J. Opt. Soc. Am. A}, 7(8):1604-1613, 1990.

\bibitem{Friberg1994} [19] A. T. Friberg, E. Tervonen, and J. Turunen. Interpretation and experimental demonstration of twisted Gaussian Schell-model beams. \textit{J. Opt. Soc. Am. A}, 11(6):1818-1826, 1994.

\bibitem{Carter1977} [20] W. H. Carter. Statistical radiometry and the theory of partial coherence. \textit{J. Opt. Soc. Am.}, 67(6):785-795, 1977.

\bibitem{Moyal1949} [21] J. E. Moyal. Quantum mechanics as a statistical theory. \textit{Proc. Cambridge Philos. Soc.}, 45:99-124, 1949.

\bibitem{Tatarskii1983} [22] V. I. Tatarskii. The Wigner representation of quantum mechanics. \textit{Sov. Phys. Usp.}, 26(4):311-327, 1983.

\bibitem{Littlejohn1986} [23] R. G. Littlejohn. The semiclassical evolution of wave packets. \textit{Phys. Rep.}, 138(4-5):193-291, 1986.

\bibitem{Visser2002} [24] T. D. Visser and E. Wolf. The origin of the Gouy phase anomaly and its generalization to astigmatic wavefields. \textit{Opt. Commun.}, 207(1-6):7-13, 2002.

\bibitem{Chaves2008} [25] J. Chaves. \textit{Introduction to Nonimaging Optics}. CRC Press, 2008.

\bibitem{McCluney1994} [26] W. R. McCluney. \textit{Introduction to Radiometry and Photometry}. Artech House, 1994.

\bibitem{Bastiaans1978} [27] M. J. Bastiaans. The Wigner distribution function applied to optical signals and systems. \textit{Opt. Commun.}, 25(1):26-30, 1978.

\bibitem{Alonso2000} [28] M. A. Alonso. Wigner functions in optics: describing beams as ray bundles and pulses as particle ensembles. \textit{Opt. Commun.}, 216(4-6):313-329, 2000.

\bibitem{Arnaud1976} [29] J. A. Arnaud. \textit{Beam and Fiber Optics}. Academic Press, 1976.

\bibitem{Twiss1958} [30] R. Q. Twiss. Radiation transfer and the possibility of negative absorption in radio astronomy. \textit{Aust. J. Phys.}, 11(4):564-579, 1958.
\end{thebibliography}
\end{document}